\documentclass[10pt]{IEEEtran}
\onecolumn
\pdfoutput=1

\usepackage[colorlinks=true,pdfstartview=FitV,linkcolor=black,citecolor=black,urlcolor=blue,plainpages=false]{hyperref}

\usepackage[numbers,sort&compress]{natbib}

\usepackage[T1]{fontenc}
\usepackage[latin1]{inputenc}

\usepackage{amsmath}
\usepackage{amsthm}
\usepackage{amsfonts}
\usepackage{dsfont}
\usepackage{mathrsfs}
\usepackage{fullpage}
\usepackage{enumerate}
\usepackage{stfloats}

\usepackage[]{authblk}
\usepackage{times}
\usepackage{epsf}
\usepackage{psfrag}
\usepackage{epsfig}

\usepackage{amssymb}
\usepackage{amstext}
\usepackage{latexsym}
\usepackage{color}
\usepackage{ifthen}
\usepackage{multirow}

\usepackage{subfigure}

\newcommand{\Z}{{\mathds Z}}

\newcommand{\ep}{\epsilon}
\newcommand{\vep}{\varepsilon}
\newcommand{\C}{{\mathcal C}}

\newcommand{\D}{{\mathcal D}}
\newcommand{\V}{{\mathcal V}}

\newcommand{\N}{{\mathcal N}}
\newcommand{\U}{{\mathcal U}}
\newcommand{\T}{{\mathcal T}}
\newcommand{\IN}{{\mathbb N}}
\newcommand{\R}{{\mathbb R}}

\newcommand{\kkk}{K \times K \times K}

\newcommand{\Ht}{{\mathcal H^{(t)}}}

\newcommand{\til}{\tilde}
\renewcommand{\ge}{\gamma}

\newcommand{\defi}{\triangleq}

\newcommand{\bv}{\bf}

\newcommand{\goesto}{\rightarrow}

\newcommand{\one}{\mathds 1}

\newcommand{\st}{:}

\newcounter{constcount}
\newcounter{numcount}
\newcommand{\bvec}[1]{{\bf #1}}

\newcounter{thmcnt}
  \let\Oldsection\section
  
\renewcommand{\section}{\stepcounter{thmcnt}\Oldsection}

\newenvironment{lemmarep}[1]{\noindent {\bf Lemma #1.}\begin{it}}{\end{it}}

\newtheorem{theorem}{Theorem}
\newtheorem{lemma}{Lemma}
\newtheorem{definition}{Definition}
\newtheorem{claim}{Claim}

\newtheorem{cor}{Corollary}

\newcounter{examplecounter}
\newcommand{\aln}[1]{\begin{align*}#1\end{align*}}

\newcommand{\al}[1]{\begin{align}#1\end{align}}

\def\Item$#1${\item $\displaystyle#1$
   \hfill\refstepcounter{equation}(\theequation)}

\begin{document}


\title{Degrees of Freedom of Two-Hop Wireless Networks:  ``Everyone Gets the Entire Cake''}


\author{Ilan Shomorony and A. Salman Avestimehr%
\thanks{I. Shomorony and A. S. Avestimehr are with the School of Electrical and Computer Engineering, Cornell University, Ithaca, NY 14853 USA (e-mails: is256@cornell.edu, avestimehr@ece.cornell.edu).}%
\thanks{Part of this paper was presented at the Allerton Conference 2012 \cite{dofkkkallerton}.}}

\maketitle

\begin{abstract}
We show that fully connected two-hop wireless networks with $K$ sources, $K$ relays and $K$ destinations have $K$ degrees of freedom both in the case of time-varying channel coefficients and in the case of constant channel coefficients (in which case the result holds for almost all values of constant channel coefficients).
Our main contribution is a new achievability scheme which we call \emph{Aligned Network Diagonalization}.
This scheme allows the data streams transmitted by the sources to undergo a diagonal linear transformation from the sources to the destinations, thus being received free of interference by their intended destination.
In addition, we extend our scheme to multi-hop networks with fully connected hops, and multi-hop networks with MIMO nodes, for which the degrees of freedom are also fully characterized.
\end{abstract}

\section{Introduction}

The conventional design of wireless networks is based on a centralized architecture where a base station, or an access point, directly exchanges data with the end users.
Thus, communication is essentially restricted to the one-to-many (broadcast) and many-to-one (multiple-access) \emph{single-hop} paradigms.
However, as the number of users and the data demand increase, and we move quickly towards the future of wireless networks, 
\emph{multi-hop} and \emph{multi-flow} paradigms are expected to play a very important role by enabling a denser spatial reuse of the spectrum and adaptation to heterogeneous scenarios characterized by user-deployed and user-operated infrastructures.

A major challenge in multi-hop multi-flow wireless networks is that ``interference management'' and ``relaying'' are coupled with each other.
In other words, wireless relay nodes must play a dual role: they serve as intermediate steps for multi-hop communication and as part of the mechanism that allows interference management schemes.
Nonetheless, in the information theory literature, these two tasks have traditionally been addressed separately.
The relaying problem is usually studied in the context of multi-hop single-flow wireless networks (or relay networks). 
For such networks, the capacity is shown in \cite{ADTJ09} to be within a constant gap to the cut-set bound, and several relaying strategies are known to achieve the capacity to within a constant gap
(e.g., quantize-map-forward \cite{ADTJ09}, lattice quantization followed by map-and-forward \cite{SuhasLatticeRelay} and compress-and-forward \cite{NoisyNetworkCoding}).
On the other hand, the problem of interference management is mostly studied in the context of multi-flow single-hop wireless network (or interference channels). 
While the exact capacity and even a constant-gap capacity approximation for interference channels are still unknown (except in the two-user case \cite{ETW,BreslerTse,BreslerTseEuro,KhandaniWIC,GerhardWIC,ElGamalCosta,SatoGInterference,VenuWIC} and some special $K$-user cases \cite{VenuWIC,SridharanKuserStrong,CarleialVS}), the total degrees of freedom of such networks are known to be half of the cut-set bound and are achievable by interference alignment techniques 
\cite{CadambeJafar,MotahariRealInterference}.  

\begin{figure}[ht] 
     \centering
       \includegraphics[height=32mm]{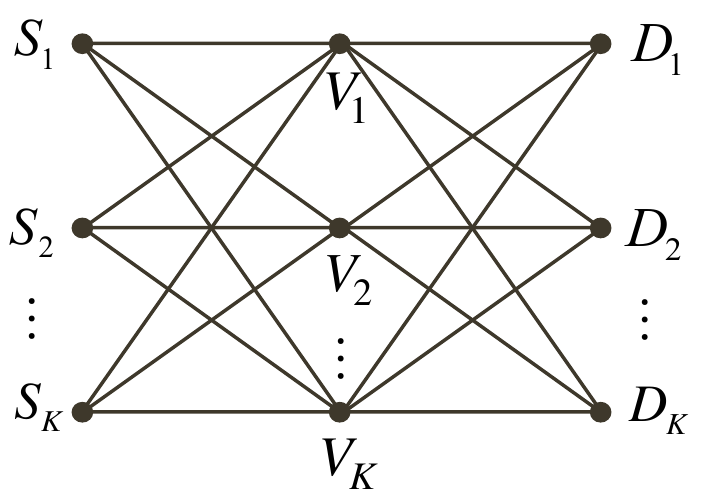} \caption{The $\kkk$ Wireless Network. \label{kkkfig}}
\end{figure}

As we move to the multi-hop multi-flow paradigm, a natural question is whether a decoupled approach for relaying and interference management is optimal.
For example, if we consider a coarse metric such as degrees of freedom, do we need coupled strategies in order to perform optimally?
To make this question clear, consider the $\kkk$ wireless network, shown in Figure \ref{kkkfig}.
One approach that decouples relaying and interference management would basically consist of viewing the $\kkk$ wireless network as the concatenation of two $K$-user interference channels.
For the $K$-user interference channel, 
it is known that, for almost all values of channel gains, $K/2$ degrees of freedom are achievable both when the channel gains are fixed and when they are time-varying \cite{CadambeJafar,MotahariRealInterference}.
Therefore, 
by repeating the scheme described in \cite{CadambeJafar,MotahariRealInterference} at each hop, 
we can also achieve $K/2$ degrees of freedom on the $\kkk$ wireless network for almost all values of the channel gains.
Another similar decoupled approach consists of viewing each hop of the $\kkk$ wireless network as a $K$-user X-channel.
This approach in fact achieves $K^2/(2K-1)$ degrees of freedom \cite{CadambeJafarX}, which is slightly better than $K/2$.
A strategy that couples relaying and interference management can be devised using the result from \cite{RankovWittneben} that shows that, in an $N \times K \times N$ wireless network, a linear scheme can neutralize the interference at all destinations as long as $K \geq N(N-1)+1$.
Thus, it is possible to achieve $\max \{ N : K \geq N(N-1)+1\}$ (roughly $\sqrt K$) degrees of freedom on the $\kkk$ wireless network, by using only a subset of $N$ source-destination pairs.
As depicted in Fig. \ref{approaches}, this coupled scheme only outperforms the Interference Channel and $X$-Channel approaches for $K = 3$.
Another coupled strategy was recently proposed for the case $K = 2$ in \cite{xx}.
The proposed scheme, named Aligned Interference Neutralization, manages to achieve the cut-set bound of two degrees of freedom, and outperforms all decoupled approaches.
However, in general, for $K>2$, all known schemes fall short of the cut-set outer bound of $K$ degrees of freedom. 
This makes the $\kkk$ wireless network a canonical example of a multi-unicast network where the gap between the state-of-the-art inner bounds and the outer bounds is very significant, and an important step in understanding how 
suboptimal decoupled approaches can be in general.


\begin{figure}[ht] 
     \centering
       \includegraphics[height=60mm]{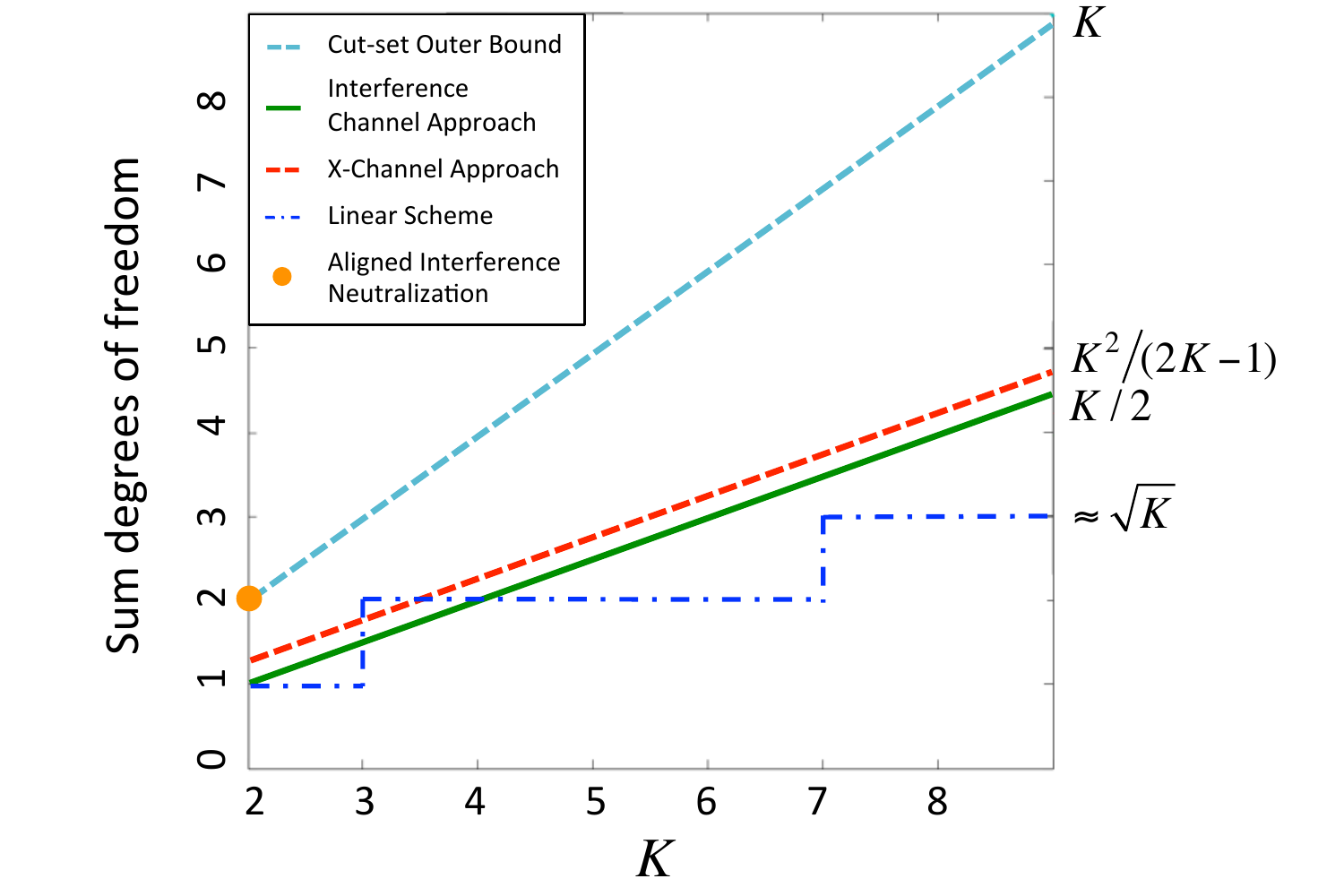} \caption{Degrees of freedom achieved by different schemes on the $\kkk$ wireless network. \label{approaches}}
\end{figure}

In this work, we introduce a new achievability scheme called \emph{Aligned Network Diagonalization} (AND), which handles relaying and interference management in a coupled manner, and manages to close the gap between inner and outer bounds.
In particular, we show that the $\kkk$ wireless network has $K$ degrees of freedom for time-varying channels and constant channels, in which case the result holds for almost all channel gain values.
This result also implies that coupled strategies can significantly outperform a strategy that handles relaying and interference management separately.

The scheme takes two forms, depending on whether the channels are fixed or time-varying. 
In the case of time-varying channels, Aligned Network Diagonalization is in fact a linear scheme.
By viewing multiple network uses as generating a single vector use of the network, we can interpret AND as a solution to a diagonalization problem: is it possible to choose linear transformations for the sources, relays and destinations such that the resulting end-to-end transformation is a diagonal transformation (with non-zero diagonal elements)? 
Our scheme shows that, with probability $1$ over the channel gain realizations, this diagonalization can indeed be obtained.
This way, interference-free channels are effectively created between each source and its corresponding destination, allowing each user to achieve arbitrarily close to one degree of freedom, i.e., each user can get ``the entire cake''.

Similar to the aligned interference neutralization scheme in \cite{xx}, each source starts by encoding its message into several data streams, each one corresponding to a direction in a vector space.
The key idea behind AND, which differentiates it from the scheme in \cite{xx}, is in the goal of the operations performed by the relays.
Each relay receives data streams along several directions, and performs carefully chosen linear operations in order to modify each of these directions.
In particular, the new directions are chosen so that \emph{it looks like the transfer matrix of the first hop is the inverse of the transfer matrix of the second hop}.
This way, by forwarding these effectively received signals, the end-to-end transformation is diagonalized.

In the case of fixed channels, however, using the network multiple times does not provide us with the diversity we need to perform the end-to-end diagonalization.
Therefore, in order to achieve the same $K$ degrees of freedom in this setting, each of the data streams is transmitted along distinct \emph{rational dimensions}, using the real interference alignment framework from \cite{MotahariRealInterference}.
Then, similar ideas to those used in the time-varying case can be used in order to modify the rational dimensions at the relays so that the transfer matrix of the first hop looks like the inverse of the transfer matrix of the second hop.
Once again the result is that the signals transmitted at the sources essentially undergo a diagonal transformation until they reach the destinations.

Several interesting extensions of our main result are possible.
In particular, for multi-hop layered networks with $K$ source-destination pairs, if all hops are fully connected, the number of degrees of freedom is the minimum between $K$ and the minimum number of relays in a single layer.
The case of MIMO sources, relays and destinations is also addressed.
Interestingly, our result implies that, from the point of view of degrees of freedom, the multiple antennas in a single MIMO relay can be equivalentely seen as separate relays, meaning that cooperation between relays in the same layer cannot increase the number of degrees of freedom.

\vspace{3mm}

\noindent \textbf{Related Work:}

Recently, a number of works have focused their attention to networks with $2$ source-destination pairs (two-unicast networks).
For instance, the work in \cite{MohajerZZ} provides constant-gap approximations to the capacity of ZZ and ZS networks.
In \cite{ThejaswiTwoHop}, the focus are $2 \times 2 \times 2$ wireless networks. 
The authors investigated how the common information between the two relays can be exploited in the second hop and proposed
relaying strategies based on distributed MIMO broadcast techniques.
In \cite{xx}, the authors also considered the $2 \times 2 \times 2$ wireless network under a degrees-of-freedom perspective.
By introducing a new scheme called aligned interference neutralization, which applies ideas from interference alignment to a multi-hop scenario, they showed that these networks have two degrees of freedom both in the case of time-varying channels and in the case of fixed channels.
General layered networks with two source-destination pairs were later considered in \cite{dof2unicastfull}.
In this work, two new notions were introduced.
The first one is the idea of network condensation, by which a network with an arbitrary number of layers is reduced to a network with at most four layers with the same degrees of freedom.
The second is a graph theoretic characterization of when the interference in a network is \emph{manageable}, i.e., when all the interference can be simultaneously neutralized.
This allowed the degrees of freedom of two-unicast layered networks with an arbitrary number of layers and arbitrary connectivity between adjacent layers to be completely characterized and shown to only attain the values $1$, $3/2$ and $2$.
In \cite{IssaTwoHop}, the authors revisited the $2 \times 2 \times 2$ setting with constant channel gains but under the constraint that the relays have to performing linear operations.
They showed that the optimal degrees of freedom in this case are $4/3$ and can be achieved by a time-varying linear scheme.

When an arbitrary number of source-destination pairs $K$ is considered, the results are scarcer.
One effort along this direction is found in \cite{SimeoneMesh}, where the authors focus on  two-hop networks   structured as $\kkk$ wireless networks where $K$ is very large (and edge effects can be neglected) and investigate communication strategies based on rate-splitting and successive interference cancellation at each hop.
In \cite{JeonDoFMulti}, networks with $K$ source-destination pairs and $K$ hops with $K$ nodes each were considered under the fast fading scenario.
The authors show that, under some assumptions on the joint distribution of the channel gains, $K$ degrees of freedom can be achieved.
The main idea is to have the relays forward their received signals at carefully chosen times, so that the signals transmitted by the sources undergo an approximately diagonal end-to-end transformation.


\section{Problem Setup} \label{setupsec}

The $\kkk$ wireless network is made up of $K$ sources $S_1,...,S_K$, $K$ relays $V_1,...,V_K$, and $K$ destinations $D_1,...,D_K$, organized as a two-hop layered network, as shown in Figure \ref{kkkfig}. 
We will consider two distinct scenarios.
\begin{itemize}
\item {\bf Time-varying channels:}\; We let the channel gain between source $S_i$ and relay $V_j$ at time $t$ be $h_{S_i,V_j}[t] \in \R$, and the channel gain between relay $V_i$ and destination $D_j$ at time $t$ be $h_{V_i,D_j}[t] \in \R$, for $t=1,2,...$.
We assume that $\{h_{S_i,V_j}[t]\}_{t=1}^\infty$ and $\{h_{V_i,D_j}[t]\}_{t=1}^\infty$ are mutually independent i.i.d~random processes each obeying an absolutely continuous probability distribution with finite second moment.
\item {\bf Constant channels:}\; We assume that $h_{S_i,V_j}[t] = h_{S_i,V_j} \in \R$ and $h_{V_i,D_j}[t]=h_{V_i,D_j} \in \R$ remain the same throughout the entire communication period.
\end{itemize}
In both cases we will assume that instantaneous channel state information is available at all nodes.
To simplify our notation, we let $\Ht = \left\{ h_{S_i,V_j}[\tau], h_{V_i,D_j}[\tau] : i,j \in \{1,...,K\}, 1 \leq \tau \leq t \right\}$ be the channel state information available at time $t$ (which includes all past channel realizations as well as the current one).
%


Communication will take place over a block of $n$ discrete time steps.
At each time $t=1,2,...,n$, each node $v \in \{S_1,...,S_K,V_1,...,V_K\}$ transmits a real-valued signal $X_v[t]$.
The received signal at a relay $V_j$ and at a destination $D_j$ are respectively given by
\al{ \label{receivedsig}
& Y_{V_j}[t] = \sum_{i=1}^K h_{S_i,V_j}[t] X_{S_i}[t] + Z_{V_j}[t]      \text{    and} \\
& Y_{D_j}[t] = \sum_{i=1}^K h_{V_i,D_j}[t] X_{V_i}[t] + Z_{D_j}[t], \label{receivedsig2}
}
where $Z_{V_j}[t]$ and $Z_{D_j}[t]$, for $t=1,2,...,n$, are sequences of i.i.d.~noise terms distributed as $\N(0,\sigma^2)$.
The noise terms are also assumed to be independent from all transmit signals and noise terms at different nodes.

\begin{definition} \label{codedef}
A coding scheme $\C$ with block length $n \in \IN$ and rate tuple $\bvec R = (R_1,...,R_{K}) \in \R^{K}$ for the $\kkk$ wireless network consists of:
\begin{enumerate}[1. ]
\item Encoding functions $f_i^{(t)} : \{1,...,2^{n R_i}\} \times \R^{tK^2} \to \R^{n}$ for each source $S_i$, $i=1,...,K$, and for each time $t=1,...,n$.
For each message $w_i \in \{1,...,2^{n R_i}\}$ and channel state information $\Ht \in \R^{t K^2}$, the codeword $f_i(w_i,\Ht)$ satisfies an average power constraint of $P$.
\item Relaying functions $r_{i}^{(t)} : \R^{t-1} \times \R^{tK^2} \to \R$, for $t=1,...,n$, for each relay $V_i$, $i=1,...,K$, satisfying the average power constraint
\aln{
\frac1n\sum_{t=1}^n \left[ r_{i}^{(t)}(y_1,...,y_{t-1},\Ht)\right]^2 \leq  P,
}
for all $(y_1,...y_{t-1}) \in \R^{t-1}$ and $\Ht \in \R^{tK^2}$.
\item A decoding function $g_i : \R^n \times \R^{n K^2} \to \{1,...,2^{n R_i}\}$ for each destination $D_i$, $i=1,...,K$.
\end{enumerate}
\end{definition}

\begin{definition}
The error probability of a coding scheme $\C$ (as defined in Definition \ref{codedef}), is given by
\aln{
P_{\rm error}(\C) = \Pr \left[ \bigcup_{i=1}^{K} \{ W_i \ne g_i(Y_{D_i}[1],...,Y_{D_i}[n]) \} \right],
}
where we assume that each $W_i$ is chosen independently and uniformly at random from $\{1,...,2^{n R_i}\}$, that source $S_i$ transmits $f_i(W_i)$ over the $n$ time-steps, and relay $V_i$ transmits $r_i^{(t)}(Y_{V_i}[1],...,Y_{V_i}[t-1])$ at time $t=1,...,n$, for $i=1,...,K$.
\end{definition}

\begin{definition} \label{achievedef}
A rate tuple $\bvec R$ is said to be achievable for the $\kkk$ wireless network if there exists a sequence of coding schemes $\C_n$ with rate tuple $\bvec R$ and blocklength $n$, for which $P_{\rm error}(\C_n) \to 0$, as $n\to \infty$.
The sequence of coding schemes $\C_n$, $n=1,2,...$, is then said to achieve rate tuple $\bvec R$.
\end{definition}

\begin{definition}
The capacity region $C(P)$ of a $\kkk$ wireless network is the closure of the set of achievable rate tuples, and the sum-capacity is defined as
\aln{
C_{\Sigma}(P) = \max_{(R_1,...,R_K) \in C(P)} \sum_{i=1}^K R_i.
}
\end{definition}

\begin{definition} \label{dofdefn}
The degrees of freedom of a $\kkk$ wireless network are defined as
\aln{
d_{\Sigma} = \lim_{P\to \infty} \frac{C_\Sigma(P)}{\tfrac12 \log P}.
}
\end{definition}

\section{Main Results}

Our main result settles the question of the number of degrees of freedom of a $\kkk$ wireless network, in both the case of time-varying and constant channel coefficients.

\begin{theorem} \label{thmtv}
For a $\kkk$ wireless network with time-varying channels, $d_\Sigma = K$.
\end{theorem}

\begin{theorem} \label{thmc}
For a $\kkk$ wireless network with constant channels, $d_\Sigma = K$ for (Lebesgue) almost all values of the channel gains.
\end{theorem}

Since the cut-set outer bound trivially implies that, in both cases, $d_\Sigma \leq K$, we only need to show that $K$ degrees of freedom are achievable.
The achievability scheme we propose for both the time-varying channel case and the constant channel case are based on interference alignment techniques.
Similar to the approach taken in \cite{xx}, in the time-varying case our alignment is performed over \emph{time dimensions}, while in the constant channel case, it is performed over \emph{rational dimensions}.
More precisely, when we have time-varying channels, the alignment is performed in the vector space created by multiple channel uses, using the framework introduced in \cite{CadambeJafar}.
In this case, our construction results in a \emph{linear scheme}, i.e., where relaying functions are restricted to linear transformations.
When the channels are constant, on the other hand, alignment over time dimensions is not feasible, and we instead use the real interference alignment frameworks introduced in \cite{MotahariRealInterference}.
%
%

In both cases, each of the $K$ sources will transmit $L$ data streams, each one along a different transmit dimension (be it time or rational).
These data streams are aligned at the relays, which allows each relay to decode approximately $L$ linear combinations of the data streams which can then be re-modulated using new transmit directions.
These new transmit directions are chosen so that all the intereference is cancelled at each destination, and the $L$ data streams from each source arrive at their intended destination along independent rational dimensions, which allows perfect decoding with high probability.
These operations guarantee that, with small probability of error, the $LK$ data streams chosen at all $K$ sources are mapped to $LK$ received directions at the destinations by a diagonal linear transformation. 
Hence, we call this scheme Aligned Network Diagonalization.


The result in Theorems \ref{thmtv} and \ref{thmc} has important consequences.
Consider a two-hop $K$-unicast wireless network where, instead of having $K$ relays, we have $A$ relays; i.e., a $K \times A \times K$ wireless network.
It is easy to see that the cut-set outer bound states that no more than $\min(K,A)$ degrees of freedom can be achieved.
Now, if $A \geq K$, we can simply ignore $A-K$ of the relays and use aligned network diagonalization to achieve $K$ degrees of freedom. 
Similarly, if $K > A$, we can ignore $K-A$ source-destination pairs, and achieve $A$ degrees of freedom.
A similar idea can be used in a multihop wireless network with $J$ layers, $K$ source-destination pairs and $A_j$ relays in the $j$th layer (hence $A_1 = A_J = K$).
If we call such a network a $K \times A_2 \times ... \times A_{J-1} \times K$ wireless network, we have the following result.

\begin{cor}
For a $K \times A_2 \times ... \times A_{J-1} \times K$ wireless network, $d_\Sigma = \min_{1 \leq j \leq J} A_j$ in the time-varying case and for almost all values of the channel gains in the constant channel case.
\end{cor}

\section{Achievability Scheme}

In this section we describe the Aligned Network Diagonalization scheme, which achieves $K$ degrees of freedom on a the $\kkk$ wireless network.
First, in Section \ref{overview}, we give a high-level overview of the scheme and describe the intuition behind it.
These ideas are then formalized in Sections \ref{descriptiontv} and \ref{descriptionc}, where we consider, respectively, the time-varying case and the constant channel case, and describe the operations performed by the sources, relays and destinations.

\subsection{Scheme Overview and Intuition} \label{overview}

In order to understand the main idea behind AND, we start by considering a different but related problem.
Suppose we have a two-hop network with $K$ sources, $K$ destinations, and a single MIMO relay with $K$ (full-duplex) antennas.
Equivalently, this setup, illustrated in Fig. \ref{mimo}, can be seen as our $\kkk$ wireless network where the $K$ relay nodes are allowed to collaborate in the computation of their transmit signals.
This new problem is clearly easier than our original problem, in the sense that any scheme for the $\kkk$ wireless network can be used to achieve the same rates on the network with a single MIMO relay node.

\begin{figure*}[ht] 
     \centering
	\includegraphics[width=0.24\linewidth]{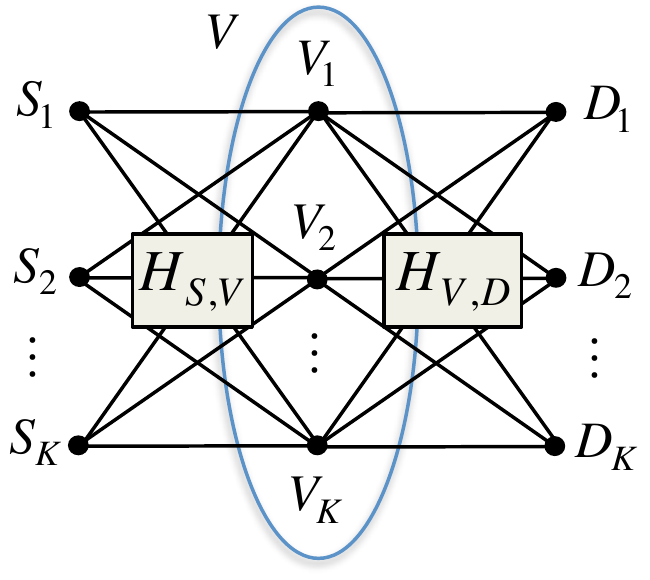} \caption{Network with single MIMO relay node. \label{mimo}}
\end{figure*}

Achieving $K$ degrees of freedom in the setting from Fig. \ref{mimo} is not difficult.
As illustrated in Fig. \ref{mimo2}, a simple scalar linear scheme can be used to \emph{diagonalize} the network.
More precisely, if each source $S_i$ transmits a signal $X_{S_i}[t]$ at time $t$, $i=1,...,K$,
the received signal at the MIMO relay at time $t$ is a length-$K$ vector $\vec Y_{V}[t] = (Y_{V_{1}}[t],...,Y_{V_{K}}[t])^{\dagger}$ given by 
$\vec Y_{V}[t] = H_{S,V} \vec X_S[t]  + \vec Z[t]$, where $\vec X_S[t] = (X_{S_1}[t],...,X_{S_K}[t])^{\dagger}$.
Then, if we assume that the transfer matrices $H_{S,V}$ and $H_{V,D}$ are invertible (which is the case with probability $1$ under the distribution assumptions in Section \ref{setupsec}), the relay can build its transmit signal for time $t+1$ through the linear transformation $\vec X_{V}[t+1] = H_{V,D}^{-1} H_{S,D}^{-1} \vec Y_{V}[t]$. 
If we let $\vec Y_{D}[t+1] = (Y_{D_{1}}[t],...,Y_{D_{K}}[t])^{\dagger}$ be the vector of the received signals at the destinations, it is clear that $\vec Y_D[t+1] = \vec X_{S}[t] + \vec {\tilde Z}[t+1]$, where $\vec  {\tilde Z}[t+1]$ is the vector of effective noises at the destinations.
Therefore, each destination receives its desired source signal plus a Gaussian noise term, meaning that the relay operations essentially diagonalized the end-to-end transfer matrix of the network, since $\vec Y_D[t+1] \approx I \vec X_{S}[t]$, where $I$ is the identity matrix.
It is easy to see that a slight modification of this scheme can guarantee that the transmit power constraints are satisfied at the relays and can thus be used to show that $K$ degrees of freedom are achievable in this setup.

\begin{figure*}[ht] 
     \centering
	\includegraphics[width=0.79\linewidth]{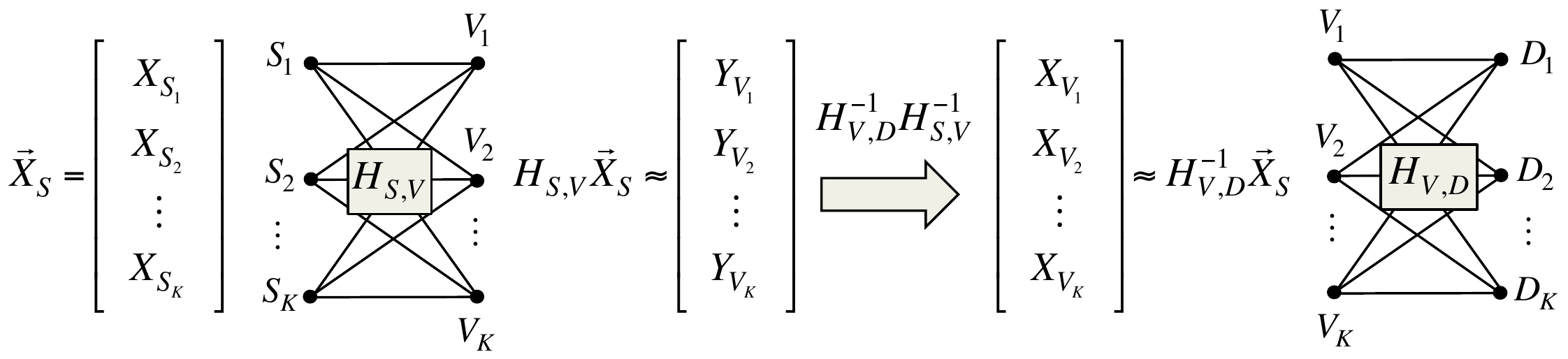} \caption{Achieving $K$ degrees of freedom on the network with sigle MIMO relay. \label{mimo2}}
\end{figure*}

When we move back to our original problem with $K$ single-antenna relay nodes, we notice that the same scheme cannot be implemented because the relays are not allowed to cooperate in order to compute $\vec X_{V}[t] = H_{V,D}^{-1} H_{S,D}^{-1} \vec Y_{V}[t]$.
Therefore, a natural question is whether it possible to apply the linear transformation $H_{V,D}^{-1} H_{S,D}^{-1}$ \emph{distributedly}.
More precisely, can we find functions $f_1,...,f_K$ such that
\al{ \label{dist}
\begin{bmatrix}
f_1(y_1) \\ f_2(y_2) \\ \vdots \\ f_K(y_K)
\end{bmatrix} 
= H_{V,D}^{-1} H_{S,D}^{-1} 
\begin{bmatrix}
y_1 \\ y_2 \\ \vdots \\ y_K
\end{bmatrix} 
}
for all $(y_1,...,y_K) \in \R^K$?
In the case of general transfer matrices $H_{V,D}$ and $H_{S,D}$, the answer is no.
In fact, if $H_{V,D}^{-1} H_{S,D}^{-1}$ is not diagonal, it is easy to see that at least one component of $H_{V,D}^{-1} H_{S,D}^{-1} (y_1,...,y_K)^{\dagger}$ depends on multiple components of $(y_1,...,y_K)$.

Therefore, in order to pursue our objective of diagonalizing the network with distributed relays, we must consider a more general question than the aforementioned one.
In particular, we will reformulate the question of whether the network can be diagonalized by bringing in the channels' time variation, and by including linear transformations at each source and at each destination.
Since our channels are time-varying, we notice that, if each hop of the network is used for $d$ consecutive time steps, we can view both the transmit signals and the received signals of the network as length-$d$ vectors.
The transfer matrix of the first hop is now given by
\aln{
H_{S,V} = \begin{bmatrix}
H_{S_1,V_1} & H_{S_2,V_1} & \cdots & H_{S_K,V_1} \\
H_{S_1,V_2} & H_{S_2,V_2} & \cdots & H_{S_K,V_2} \\
\vdots & \vdots & \ddots & \vdots \\
H_{S_1,V_K} & H_{S_2,V_K} & \cdots & H_{S_K,V_K}
\end{bmatrix}, 
\quad \text{where} \quad H_{S_i,V_j} =
\begin{bmatrix}
h_{S_i,V_j}[0] & 0 & \cdots & 0 \\
0 & h_{S_i,V_j}[1] & \cdots & 0 \\
\vdots & \vdots & \ddots & \vdots \\
0 & 0 & \cdots & h_{S_i,V_j}[d-1]
\end{bmatrix},
}
for $1 \leq i,j \leq K$.
The transfer matrix of the second hop, $H_{V,D}$ can be similarly built.
In this new setting, we have transfer matrices constituted of diagonal blocks, and we could restate the goal in (\ref{dist}) by having each $y_i$ be a length-$d$ column vector.
In this new setting, by assuming that each relay $V_i$ applies a linear transformation to its vector of $d$ received signals, the diagonalization problem becomes the problem of finding block diagonal matrices $A_V$ (with $d\times d$ blocks $A_{V_i}$, for $i=1,...,K$), $A_S$ (with $d\times d'$ blocks $A_{S_i}$, for $i=1,...,K$) and $A_D$ (with $d'\times d$ blocks $A_{D_i}$, for $i=1,...,K$) such that
\al{ \label{diag}
A_D H_{V,D} A_V H_{S,V} A_S = I,
}
where $A_S \in \R^{Kd \times Kd'}$, $A_V \in \R^{Kd \times Kd}$ and $A_D \in \R^{Kd' \times Kd}$ correspond to the linear transformations applied by the sources, relays and destinations.
Notice that the identity matrix $I$ is $Kd' \times Kd'$, and the parameter $d'$ regulates how much information the sources are transmitting.
Our goal is to solve the problem specified by (\ref{diag}) for $d' \leq d$ large enough so that $d'/d \approx 1$.

In this work, our main contribution is to show that the problem in (\ref{diag}), with probability $1$ over the channel realizations, indeed admits a sequence of solutions parameterized by $d$, with the property that $d'/d \to 1$ as $d\to \infty$.
The scheme that provides this solution, which we call Aligned Network Diagonalization, can be roughly described as follows.
The source matrices $A_{S_i}$, $i=1,...,K$, are all chosen to be the same $d \times d'$ matrix $A_{S_0}$, whose columns are all of the form 
\al{ \label{col}
T_{s_{1,1},s_{1,2},...,s_{K,K}} = \prod_{1\leq i,j \leq K} H_{S_i,V_j}^{s_{i,j}} \one,
}
for some nonnegative integers $s_{i,j}$, $1 \leq i,j \leq K$, where $\one$ is a column vector with all entries equal to $1$.
It is then not difficult to see that the result of 
\aln{
H_{S,V} A_S = \begin{bmatrix}
H_{S_1,V_1} & \cdots & H_{S_K,V_1} \\
\vdots  & \ddots & \vdots \\
H_{S_1,V_K} & \cdots & H_{S_K,V_K}
\end{bmatrix}
\begin{bmatrix}
A_{S_0} & \cdots & 0 \\
\vdots  & \ddots & \vdots \\
0 & \cdots & A_{S_0}
\end{bmatrix}
}
is a $Kd \times Kd'$ matrix with $d \times d'$ blocks whose columns are again of the form in (\ref{col}).
The key idea in the AND scheme is in the design of the relaying matrices $A_{V_i}$.
Once again, we will choose a single matrix $A_{V_0}$ and let $A_{V_i} = A_{V_0}$ = $\tilde {\bv T} \, {\bv T}^{-1}$ for $i=1,...,K$, where
$\bv T$ is a matrix whose columns are the vectors of the form (\ref{col}) that appear in any of the blocks in $H_{S,V} A_S$ 
and $\tilde {\bv T}$ is obtained from ${\bv T}$ by replacing each column $T_{s_{1,1},s_{1,2},...,s_{K,K}}$ as given in (\ref{col}) with the column 
\aln{
\tilde T_{s_{1,1},s_{1,2},...,s_{K,K}} = \prod_{1\leq i,j \leq K} B_{i,j}^{s_{i,j}} \one,
}
for diagonal matrices $B_{i,j}$ to be defined.
The key observation is that the result of any vector $T_{s_{1,1},s_{1,2},...,s_{K,K}}$, as given in (\ref{col}), undergoing the transformation $A_{V_0}$ is 
\aln{
\tilde {\bv T} \, {\bv T}^{-1} T_{s_{1,1},s_{1,2},...,s_{K,K}} & = \tilde {\bv T} \, {\bv T}^{-1} \, {\bv T} \; {\bv e}_{s_{1,1},s_{1,2},...,s_{K,K}} = \tilde  T_{s_{1,1},s_{1,2},...,s_{K,K}},
}
where $ {\bv e}_{s_{1,1},s_{1,2},...,s_{K,K}}$ is a standard basis vector with the $1$ at the entry corresponding to the position of the column $T_{s_{1,1},s_{1,2},...,s_{K,K}}$ in $\bv T$.
Therefore, the transformation $A_{V_0}$ applied by each relay can be understood as replacing each ``direction'' $T_{s_{1,1},s_{1,2},...,s_{K,K}}$ with a new direction $\tilde T_{s_{1,1},s_{1,2},...,s_{K,K}}$.
Each matrix $B_{i,j}$ is chosen as \emph{what $T_{s_{1,1},s_{1,2},...,s_{K,K}}$ would have been if $H_{S,V} = H_{V,D}^{-1}$}.
This essentially makes it look like the first hop of the network is $H_{V,D}^{-1}$, rather than $H_{S,V}$. More precisely, we have
$A_V H_{S,V} A_S = H_{V,D}^{-1} \tilde A_S$, where $\tilde A_S$ is obtained by replacing each column $T_{s_{1,1},s_{1,2},...,s_{K,K}}$ in one of the blocks of $A_S$ with $\tilde T_{s_{1,1},s_{1,2},...,s_{K,K}}$.
This reduces the end-to-end transformation in (\ref{diag}) to
\aln{
A_D H_{V,D} A_V H_{S,V} A_S =  A_D H_{V,D} H_{V,D}^{-1} \tilde A_S = A_D \tilde A_S.
}
Finally, since $\tilde A_S$ can be seen to admit a block diagonal left inverse, we can set $A_D$ to be this matrix and obtain our desired end-to-end diagonalization.
In the next section, we describe this scheme in more detail.
In particular, several issues such as power constraints and invertibility of the matrices are properly addressed, and the fact that we can choose $d'$ and $d$ sufficiently large such that $d'/d$ approaches $1$ is proved.

\subsection{Aligned Network Diagonalization for Time-Varying Channels} \label{descriptiontv}


In order to use the Aligned Network Diagonalization in the time-varying scenario, sources and relays will choose their transmit directions based on the channel gain values at each time-step.

\vspace{2mm}

\noindent {\bf Encoding at the sources:} 

\vspace{2mm}

\noindent Each source $S_i$ starts by breaking its message $W_i$ into $L$ submessages.
Each of the submessages will be encoded in a separate data stream, using Gaussian random codebooks with codewords of length $n$ and entries drawn as $\N(0,P)$.
We will let 
\aln{ 
T_{s_{11},s_{12},...,s_{KK}}[t] = \prod_{\substack{1 \leq i \leq K \\ 1 \leq j \leq K} } h_{S_i,V_j}[t]^{s_{ij}},
}
and $\Delta_N = \{0,...,N-1\}^{K^2}$, 
and we define the set of transmit directions for the sources at time $t$ to be
\al{ \label{transdirections0}
\T_N[t] = \left\{  T_{s_{11},s_{12},...,s_{KK}}[t] : (s_{11},s_{12},...,s_{KK}) \in \Delta_N \right\},
}
for some arbitrary $N$.
This selection of directions is similar in flavor to the directions chosen in the Interference Alignment scheme introduced in \cite{CadambeJafar}.
Notice that the number of transmit directions (which is also the number of data streams) is $L = |\T_N[t]| = |\Delta_N| = N^{K^2}$.
To simplify the notation we will let $\vec s$ be a vector of indices $(s_{11},s_{12},...,s_{KK})$ and write $T_{\vec s}$.

Communication will take place over a block of $n d$ time-steps, where $d \defi (N+1)^{K^2}$.
The $(m+1)$th symbol of the codeword associated to the submessage of stream $\vec s \in \Delta_N$ of source $S_i$ will be written as $c_{i,\vec s}[m]$, for $0 \leq m \leq n-1$.
%
At time $t = md + j$ for $m \in \{0,...,n-1\}$ and $j \in \{0,...,d-1\}$, source $S_i$ will thus transmit 
\aln{
X_{S_i}[t] = \gamma \sum_{\vec s \in \Delta_N}  T_{\vec s}[t] \; c_{i,\vec s} [m].
} 
The constant $\gamma$ is chosen so that the transmit power
\al{
E\left[ X_{S_i}[t] ^2\right] & = \gamma^2 E \left[ \left(\sum_{\vec s \in \Delta_N}  T_{\vec s}[t] \; c_{i,\vec s} [m]\right)^2 \right]  \nonumber \\
& = \gamma^2 P \sum_{\vec s \in \Delta_N}   E \left[ T_{\vec s}[t]^2 \right]  \label{powerbound1}
}
does not exceed $P$.
In (\ref{powerbound1}), we used the fact that the $c_{i,\vec s}$ were independently generated.
Notice that $\gamma$ does not depend on $P$ or $t$ and can be chosen strictly positive, since the fact that the channel gains are independent and have finite variances implies $E \left[ T_{\vec s}[t]^2 \right] < \infty$ for all $\vec s$.

\vspace{4mm}


{\noindent \bf Relaying operations:} 

\vspace{2mm}

\noindent The received signal at relay $V_j$ at time $t=md+j$ can be written as
\al{ \label{recnotaligned0}
Y_{V_j}[t] = \gamma \sum_{\vec s \in \Delta_{N}}  T_{\vec s}[t] \left( \sum_{i=1}^K h_{S_i,V_j}[t] c_{i,\vec s}[m] \right) + Z_{V_j}[t].
}
Even though writing the received signal as in (\ref{recnotaligned0}) does not emphasize the alignment that occurs at the relays, it will still be a useful representation of the received signal.
To capture the alignment, we consider rearranging the terms in the summation in (\ref{recnotaligned0}) by viewing it as a polynomial on the variables $h_{S_i,V_j}[t]$, for $1 \leq i,j \leq K$, where the coefficients are given by sums of $c_{i,\vec s}$ terms.
It can then be seen that the actual set of received directions at each relay is a subset of $\T_{N+1}[t]$, 
and the received signal at relay $V_j$ at time $t$ can be alternatively written as
\al{ \label{recaligned0}
Y_{V_j}[t] = \gamma \sum_{\vec s \in \Delta_{N+1}}  T_{\vec s}[t] \; u_{j,\vec s} [m] + Z_{V_j}[t],
}
where $u_{j,(s_{11},s_{12},...,s_{KK})}[m] = \sum_{i=1}^K c_{i,(s_{11},s_{12},...,s_{ij}-1,...,s_{KK})}[m]$ and we define $c_{i,\vec s}[m] = 0$ if any component of $\vec s$ is $-1$ or $N$.
At the end of the $(m+1)$th block of $d$ received signals (i.e., the block consisting of signals received at $t=md,md+1,...,(m+1)d-1$), relay $V_j$ can form a $d$-dimensional vector of received signals
\al{ \label{recalignedvec0}
\vec Y_{V_j}[m] = \begin{bmatrix} Y_{V_j}[md] \\ Y_{V_j}[md+1] \\ \vdots \\ Y_{V_j}[(m+1)d-1] \end{bmatrix}
= \gamma \sum_{\vec s \in \Delta_{N+1}}  
\begin{bmatrix} T_{\vec s}[md] \\ T_{\vec s}[md+1] \\ \vdots \\ T_{\vec s}[(m+1)d-1] \end{bmatrix}  \; 
u_{j,\vec s} [m] + \begin{bmatrix} Z_{V_j}[md] \\ Z_{V_j}[md+1] \\ \vdots \\ Z_{V_j}[(m+1)d-1] \end{bmatrix}
}
for $m \in \{0,...,n-1\}$.
Notice that, for each $\vec s \in \Delta_{N+1}$, $T_{\vec s}[t]$ is a distinct monomial on the variables $h_{S_i,V_j}[t]$ for $i,j \in \{1,...,K\}$.
The following lemma, whose proof is in Appendix \ref{appinvertible}, will thus be useful.


\begin{lemma} \label{invertiblelemma}
Let $\vec p (x_1,...,x_K) = \left[ p_1 (x_1,...,x_K),...,  p_d (x_1,...,x_K) \right]^\dagger$, where each $p_i(x_1,...,x_K)$ is a distinct monomial on the variables $x_1,...,x_K$.
Then, the determinant of the $d \times d$ matrix
\aln{
\left[ \vec p (x_{1,1},...,x_{1,K}), \vec p (x_{2,1},...,x_{2,K}), ..., \vec p (x_{d,1},...,x_{d,K}) \right]
}
is a non-identically zero polynomial on the variables $x_{1,1},...,x_{1,K},...,x_{d,1},...,x_{d,K}$.
\end{lemma}

Let ${\bv T}[m]$ be the $d \times d$ matrix whose columns are
\al{ \label{vect}
\vec T_{\vec s}[m] = \begin{bmatrix} T_{\vec s}[md] \\ T_{\vec s}[md+1] \\ \vdots \\ T_{\vec s}[(m+1)d-1] \end{bmatrix},
}
for $\vec s \in \Delta_{N+1}$.
From Lemma \ref{invertiblelemma}, we see that $\det {\bf T}[m]$, seen as a polynomial on the variables $h_{S_i,V_j}[t]$ for $i,j \in \{1,...,K\}$ and $t=md,...,(m+1)d-1$, is not identically zero.
Thus, since $h_{S_i,V_j}[t]$ for $i,j \in \{1,...,K\}$ and $t=0,...,nd-1$ are all indepedent and drawn from absolutely continuous distributions, ${\bf T}[m]$ is invertible with probability $1$.
Moreover, if we fix some arbitrary $\ep > 0$, we can find $\delta > 0$ such that $|\det {\bf T}[m]| > \delta$ with probability $1-\ep$. 
At time $t=(m+1)d-1$, the relays will verify whether this is satisfied. 
In case $|\det {\bf T}[m]| \leq \delta$, all the relays will simply remain silent at times $t=(m+1)d,...,(m+2)d-1$.
As we will see later, this is important to guarantee that the entries of ${\bf T}^{-1}$ are not too large, which could lead to a violation of the transmit power constraints at the relays.
Otherwise, if $|\det {\bf T}[m]| > \delta$, in order to build its transmit signals,
each relay $V_j$ will construct the vector of estimates of the $u_{j,\vec s}$s
\al{ \label{uestimate} 
\begin{bmatrix} \hat u_{j,\vec s}[m] \end{bmatrix}_{\vec s \in \Delta_{N+1}}
= \gamma^{-1} {\bv T}[m]^{-1} \vec Y_{V_j}[m]
= \begin{bmatrix} u_{j,\vec s}[m] \end{bmatrix}_{\vec s \in \Delta_{N+1}} + \gamma^{-1} {\bv T}[m]^{-1} \begin{bmatrix} Z_{V_j}[md] \\ Z_{V_j}[md+1] \\ \vdots \\ Z_{V_j}[(m+1)d-1] \end{bmatrix}. 
}
In order to build the transmit signal for time $t = (m+1)d,...,(m+2)d-1$, each relay will compute the determinant of  
\aln{
{\bf H}_{V,D}[t] = 
\begin{bmatrix} h_{V_1,D_1}[t] & ... & h_{V_K,D_1}[t] \\ \vdots & \ddots & \vdots \\ 
h_{V_1,D_K}[t] & ... & h_{V_K,D_K}[t] \end{bmatrix}.
}
Lemma \ref{invertiblelemma} in this case implies that $\det {\bf H}_{V,D}[t]$ is a non-identically zero polynomial on the variables $h_{V_i,D_j}[t]$, $i,j \in \{1,...,K\}$, and we can find $\delta' > 0$ such that $|\det{\bf H}_{V,D}[t]| > \delta'$ for $t=md,...,(m+1)d-1$ with probability $1- \ep'$, for any fixed $\ep' >0$.
Since the event $\{|\det{\bf H}_{V,D}[t]| > \delta' \}$ is independent for each time $t$, we will choose $\ep'$ and the corresponding $\delta'$ small enough so that 
\al{ \label{probnotenough}
\Pr\left[ \Big| \left\{ t : md \leq t \leq (m+1)d-1, \,|\det{\bf H}_{V,D}[t]| > \delta^{\prime} \right\} \Big| \leq |\Delta_N| \right] < \ep,
}
where $\ep$ is the same previously chosen parameter.
%
%
If $|\det{\bf H}_{V,D}[t]| \leq \delta'$, all relays will simply stay silent at time $t$.
Otherwise, after obtaining $\begin{bmatrix} \hat u_{j,\vec s}[m] \end{bmatrix}_{\vec s \in \Delta_{N+1}}$, relay $V_j$ will encode all these $d = |\Delta_{N+1}|$ symbols using new transmit directions.
To describe the new set of transmit directions, we first define
\al{ \label{bdef}
\begin{bmatrix} b_{11}[t] & ... & b_{K1}[t] \\ \vdots & \ddots & \vdots \\ 
b_{1K}[t] & ... & b_{KK}[t] \end{bmatrix} = {\bf H}_{V,D}[t]^{-1}.
}
Next, we let
\al{ \label{directionsfixed2}
\til T_{s_{11},s_{12},...,s_{KK}}[t] = \prod_{\substack{1 \leq i \leq K \\ 1 \leq j \leq K} } b_{ij}[t]^{s_{ij}},
}
and, similar to (\ref{transdirections0}), we can define the set of transmit directions for the relays to be
\al{ \label{tiltdef}
\til \T_{N+1}[t] = \left\{  \til T_{s_{11},s_{12},...,s_{KK}}[t] : (s_{11},s_{12},...,s_{KK}) \in \Delta_{N+1} \right\}.
}
Relay $V_j$ will encode the $\hat u_{j,\vec s}$\,s by transmitting, at time $t = (m+1)d,(m+1)d+1,...,(m+2)d-1$,
\al{ \label{transaligned0} 
X_{V_j}[t] = \gamma' \left( \sum_{\vec s \in \Delta_{N+1}}  \til T_{\vec s}[t] \; \hat u_{j,\vec s} [m] \right)
= \gamma' \left( \sum_{\vec s \in \Delta_{N+1}}  \til T_{\vec s}[t] \; u_{j,\vec s} [m] \right) + \til Z_{V_j}[t],
}
where $\til Z_{V_j}[t]$ is the effective noise term which results from the additive noise terms in the estimates $\hat u_{j,\vec s}$s.
The constant $\gamma'$ is chosen so that the transmit power
\aln{
E\left[ X_{V_j}[t] ^2\right] & = {\gamma'}^2 E \left[ \left(\sum_{\vec s \in \Delta_{N+1}} {\tilde T}_{\vec s}[t] \; u_{j,\vec s} [m]\right)^2 \right]  + 
 E\left[ \tilde Z_{V_j}[t]^2 \right]    \\
& \leq {\gamma'}^2 K P \sum_{\vec s \in \Delta_{N+1}}   E \left[ {\tilde T}_{\vec s}[t]^2 \right]  + E\left[ \tilde Z_{V_j}[t]^2 \right]
}
does not exceed $P$.
Since each $b_{ij}[t]$ can be written as a ratio between a polynomial on the variables $h_{V_i,D_j}[t]$, $i,j \in \{1,...,K\}$ and $\det {\bv H}_{V,D}[t]$, and $|\det {\bv H}_{V,D}[t]| > \delta'$, we see that $E \left[ \tilde  T_{\vec s}[t]^2 \right] < \infty$ for all $\vec s$.
Moreover, the fact that $E\left[ h_{V_i,D_j}[t]^2 \right] < \infty$, for each $i,j \in \{1,...,K\}$, and $|\det {\bv T}[m]| > \delta$ guarantees that the variance of $\til Z_{V_j}[t]$ is finite and independent of $P$.
Thus, for $P$ sufficiently large, $\gamma'$ can be chosen independent of $P$ and $t$.


We then have the following claim.

\begin{claim} \label{claim1}
The transmit signal of relay $V_j$, given in (\ref{transaligned0}), can be re-written as
\al{ \label{transnotaligned0}
X_{V_j}[t] = \gamma' \sum_{\vec s \in \Delta_{N}}  \til T_{\vec s}[t] \left( \sum_{i=1}^K b_{ij}[t] \, c_{i,\vec s} [m] \right) + \tilde Z_{V_j}[t].
}
\end{claim}

\begin{proof}[Proof]
The main idea is to notice that, just as (\ref{recaligned0}) can be written as (\ref{recnotaligned0}), (\ref{transaligned0}) can be re-written as (\ref{transnotaligned0}).
This can be more easily understood if we think of the (noiseless version of the) received signal in (\ref{recaligned0}) as a polynomial on the variables $h_{S_i,V_j}[t]$, $1 \leq i,j \leq K$. 
When relay $V_j$ estimates each coefficient $u_{j,\vec s}[t]$ of this polynomial and then 
replaces each monomial $T_{\vec s}$ with $\til T_{\vec s}$, it is essentially re-building the same polynomial
with each variable $h_{S_i,V_j}[t]$ replaced by $b_{ij}[t]$.
Therefore, the same factorization used on the polynomial on the $h_{S_i,V_j}[t]$ variables in (\ref{recnotaligned0}) can be used on the polynomial on the $b_{ij}[t]$ variables, as shown in (\ref{transnotaligned0}).
\end{proof}

\vspace{4mm}

{\noindent \bf Decoding at the destinations:} 

\vspace{2mm}

\noindent In order to compute the received signals at the destinations, we first notice that, from (\ref{transnotaligned0}), the vector of the $K$ relay transmit signals at time $t \in \{ (m+1)d,(m+1)d+1,...,(m+2)d-1\}$, can be written as
\al{ \label{relaytransmitvector}
\begin{bmatrix} X_{V_1}[t]  \\ \vdots \\ X_{V_K}[t] \end{bmatrix} = \gamma' \sum_{\vec s \in \Delta_{N}}  \til T_{\vec s}[t]
\begin{bmatrix} b_{11}[t] &  ... & b_{K1}[t] \\  \vdots & \ddots & \vdots \\ 
b_{1K}[t] & ... & b_{KK}[t] \end{bmatrix}
\begin{bmatrix} c_{1,\vec s}[m] \\ \vdots \\ c_{K,\vec s}[m] \end{bmatrix}
+ \begin{bmatrix} \til Z_{V_1}[t] \\ \vdots \\ \til Z_{V_K}[t] \end{bmatrix}.
}
We can then write the vector of the $K$ received signals at the destinations as
\al{
\begin{bmatrix} Y_{D_1}[t] \\ \vdots \\ Y_{D_K}[t] \end{bmatrix} & = 
\begin{bmatrix} h_{V_1,D_1}[t] & ... & h_{V_K,D_1}[t] \\ \vdots & \ddots & \vdots \\ 
h_{V_1,D_K}[t] & ... & h_{V_K,D_K}[t] \end{bmatrix} 
\begin{bmatrix} X_{V_1}[t] \\ \vdots \\ X_{V_K}[t] \end{bmatrix} + \begin{bmatrix} Z_{D_1}[t] \\ \vdots \\ Z_{D_K}[t] \end{bmatrix} \nonumber \\ 
& = 
\begin{bmatrix} b_{11}[t] & ... & b_{K1}[t] \\ \vdots & \ddots & \vdots \\ 
b_{1K}[t] & ... & b_{KK}[t] \end{bmatrix}^{-1}
\begin{bmatrix} X_{V_1}[t] \\ \vdots \\ X_{V_K}[t] \end{bmatrix} + \begin{bmatrix} Z_{D_1}[t] \\ \vdots \\ Z_{D_K}[t] \end{bmatrix} \nonumber \\
& = \gamma' \sum_{\vec s \in \Delta_{N}}  \til T_{\vec s}[t]
\begin{bmatrix} c_{1,\vec s}[m] \\ \vdots \\ c_{K,\vec s}[m] \end{bmatrix} + \underbrace{\begin{bmatrix} b_{11}[t] & ... & b_{K1}[t] \\ \vdots & \ddots & \vdots \\ 
b_{1K}[t] & ... & b_{KK}[t] \end{bmatrix}^{-1} \begin{bmatrix} \til Z_{V_1}[t] \\ \vdots \\ \til Z_{V_K}[t] \end{bmatrix} +\begin{bmatrix} Z_{D_1}[t] \\ \vdots \\ Z_{D_K}[t] \end{bmatrix}}_{\left[\til Z_{D_1}[t] \; ... \; \til Z_{D_K}[t] \right]^{\dagger}}. \label{matrixrec0}
}
Thus, the received signal at destination $D_j$ at time $t=(m+1)d,(m+1)d+1,...,(m+2)d-1$ is simply given by
\al{ \label{destreceivedscalar0}
Y_{D_j}[t] = \gamma' \sum_{\vec s \in \Delta_{N}}  \til T_{\vec s}[t] \; c_{j,\vec s}[m] + \tilde Z_{D_j}[t],
}
and we see that all the interference has been cancelled, and destination $D_j$ receives only the data streams originated at source $S_j$.
Moreover, the effective noise $\til Z_{D_j}[t]$ has a finite variance that is independent of $P$.

Destination $D_j$ will use decoding operations similar to those used at the relays.
The block of the $d$ signals received at times $t=(m+1)d,(m+1)d+1,...,(m+2)d-1$ can be used to form a length-$d$ vector
\al{ \label{recalignedvec}
\vec Y_{D_j}[m+1] = \begin{bmatrix} Y_{D_j}[(m+1)d] \\ Y_{D_j}[(m+1)d+1] \\ \vdots \\ Y_{D_j}[(m+2)d-1] \end{bmatrix}.
}
Notice that, in case $|\det {\bf T}[m]| \leq \delta$, these received signals will contain just noise, since the relays stayed silent in times $t=(m+1)d,(m+1)d+1,...,(m+2)d-1$.
Moreover, at any time $t \in \{(m+1)d,(m+1)d+1,...,(m+2)d-1\}$ for which $|\det{\bf H}_{V,D}[t]| \leq \delta'$, the corresponding entry of $\vec Y_{D_j}[m+1]$ will contain only noise.
Notice that, from (\ref{probnotenough}), this will be the case of less than $d - |\Delta_N|$ of the entries, with probability at least $1- \ep$.
Thus, with probability at least $1-\ep$, destination $D_j$ can let $t_1,..., t_{|\Delta_N|}$ be the first $|\Delta_N|$ values of $t \in \{(m+1)d,...,(m+2)d-1\}$ for which $|\det{\bf H}_{V,D}[t]| > \delta'$, and, from (\ref{destreceivedscalar0}), the resulting received signals satisfy
\al{ \label{recalignedvec}
\begin{bmatrix} Y_{D_j}[t_1] \\ Y_{D_j}[t_2] \\ \vdots \\ Y_{D_j}[t_{|\Delta_N|}] \end{bmatrix}
= \gamma' \sum_{\vec s \in \Delta_{N}}  
\begin{bmatrix} \til T_{\vec s}[t_1] \\ \til T_{\vec s}[t_2] \\ \vdots \\ \til T_{\vec s}[t_{|\Delta_N|}] \end{bmatrix}  \; 
c_{j,\vec s} [m] + \begin{bmatrix} \til Z_{D_j}[t_1] \\ \til Z_{D_j}[t_2] \\ \vdots \\ \til Z_{D_j}[t_{|\Delta_N|}] \end{bmatrix}.
}
%
%
%
%
%
We will let
$
\vec {\til T}_{\vec s}[m+1] = \begin{bmatrix} \til T_{\vec s}[t_1], \til T_{\vec s}[t_2],...,\til T_{\vec s}[t_{|\Delta_N|}] \end{bmatrix}^\dagger,
$
for each $\vec s \in \Delta_N$.
Notice that, for each $\vec s \in \Delta_{N}$, $\til T_{\vec s}[t]$ is a distinct monomial on the variables $h_{V_i,D_j}[t]$ for $i,j \in \{1,...,K\}$.
We will then let $\bv {\til T} [m]$ be the $|\Delta_N| \times |\Delta_N|$ (i.e., $N^{K^2} \times N^{K^2}$) matrix whose columns are $\vec {\til T}_{\vec s}[m+1]$,
for $\vec s \in \Delta_{N}$.
The remaining $d - |\Delta_N|$ received signals 
are simply discarded by the destinations.
Lemma \ref{invertiblelemma} now implies that $\det \bv {\til T} [m]$ is a non-identically zero polynomial on the variables $h_{V_i,D_j}$, $i,j \in \{1,...,K\}$, and ${\bf \til T} [m]$ is invertible with probability $1$.
Moreover, one can find $\delta''>0$ such that, with probability $1-\ep$, $|\det {\bf \til T} [m]| > \delta''$.

At time $t = (m+2)d-1$, destination $D_j$ will construct a length-$|\Delta_N|$ vector of effective outputs as follows.
If $|\det {\bf \til T} [m]| \leq \delta''$, $|\det {\bf T} [m]| \leq \delta$ or if there are more than $d-|\Delta_N|$ times $t \in \{(m+1)d,(m+1)d+1,...,(m+2)d-1\}$ for which $|\det{\bf H}_{V,D}[t]| \leq \delta'$, it simply outputs $[ \vep, ..., \vep ]$, where $\vep$ simbolizes an erasure.
Since each of these three events occurs with probability at most $\ep$, their union occurs with probability at most $3 \ep$.
If none of these events occurs,
%
%
%
%
%
destination $D_j$ will output the vector of estimates of the $c_{j,\vec s}$s
\al{ 
\begin{bmatrix} \hat c_{j,\vec s}[m] \end{bmatrix}_{\vec s \in \Delta_{N}}
& = \frac{1}{\gamma'} {\bf \til T}[m]^{-1} \begin{bmatrix} Y_{D_j}[t_1] \\ Y_{D_j}[t_2] \\ \vdots \\ Y_{D_j}[t_{|\Delta_N|}] \end{bmatrix} \nonumber \\
& = \begin{bmatrix} c_{j,\vec s}[m] \end{bmatrix}_{\vec s \in \Delta_{N}} + \underbrace{\frac{1}{\gamma'} {\bf \til T}[m]^{-1} \begin{bmatrix} \til Z_{V_j}[t_1] \\ \til Z_{V_j}[t_2] \\ \vdots \\ \til Z_{V_j}[t_{|\Delta_N|}] \end{bmatrix}}_{\left[ \tilde{\tilde{Z}}_{j,\vec s}[m]\right]_{\vec s \in \Delta_N}}. \label{cestimate}
}
Notice that $|\det {\bf \til T} [m]| > \delta''$ implies that the entries of ${\bf \til T} [m]^{-1}$ have a finite variance, which in turn implies that the resulting additive noise vector $\left[ \tilde{\tilde{Z}}_{j,\vec s}[m]\right]_{\vec s \in \Delta_N}$ has a finite covariance matrix, with entries that are independent of $P$.
Destination $D_j$ will then view each entry of its output vector as the output of a separate channel, with input $c_{j,\vec s}[m]$ and output $c_{j,\vec s}[m] + \tilde{\tilde{Z}}_{j,\vec s}[m]$ with probability $1-q$ and $\vep$ with probability $q$, where $q \leq 3 \ep$.
Therefore, we essentially create $N^{K^2}$ parallel AWGN channels with erasure probability at most $3 \ep$.
The fact that the additive noises are correlated is irrelevant (in fact it can only improve the achievable rates), and it is clear that we can achieve $1-3 \ep$ degrees of freedom in each of these effective channels.
Since we need $d$ time-steps to transmit one symbol in each of these channels, we achieve a total of
\aln{
(1-3\ep)\frac{N^{K^2}}{d} = (1-3\ep)\frac{N^{K^2}}{(N+1)^{K^2}} = (1-3\ep)\left(\frac{N}{N+1}\right)^{K^2},
}
for arbitrarily chosen $N$ and $\ep$.
Thus, by choosing $N$ large and $\ep > 0$ small, each user can achieve arbitrarily close to one degree of freedom.

%
%

%

\subsection{Aligned Network Diagonalization Scheme for Constant Channels} \label{descriptionc}

In the case of constant channel gains, the AND scheme presented in Section \ref{descriptiontv} does not work.
The lack of time diversity makes the entries in the vector $\vec{T}_{\vec s}[m]$, given in (\ref{vect})), be all equal, and ${\bv T}[m]$ is not invertible (as its rank is one).
Therefore, in order to achieve the $K$ degrees of freedom with constant channels, we must perform the alignment operations of AND not over time dimensions, but over rational dimensions, in the spirit of \cite{MotahariRealInterference}. 

As in Section \ref{descriptiontv}, in this section, we describe the scheme by first considering the encoding at the sources, followed by the relaying operations and the decoding operation.
Then, in Section \ref{analysissec}, we present a performance analysis on the scheme, where we formally prove that it achieves arbitrarily close to $K$ degrees of freedom for almost all values of the channel gains.

\vspace{2mm}

\noindent {\bf Encoding at the sources:} 

\vspace{2mm}

Each source $S_i$ again starts by breaking its message $W_i$ into $L$ submessages.
Each of the submessages will be encoded in a separate data stream, using a single codebook with codewords of length $n$, 
obtained by uniform i.i.d.~sampling of the set
\al{
\U = \Z \cap \left[ - P^{\frac{1-\ep}{2(d+\ep)}} , P^{\frac{1-\ep}{2(d+\ep)}} \right], \label{codeset}
}
for a small $\ep > 0$, 
and $d = (N+1)^{K^2}$. 
The rate of this code, i.e., the number of codewords, will be determined later. 
Notice that $d$ can be thought of as a parameter which sets the number of degrees of freedom given to each stream to be $(1-\ep)/(d+\ep) \approx 1/d$.
The set of transmit directions $\T_N$ is defined as in (\ref{transdirections0}), the only difference being that we drop the time index $t$, since the channels are constant.
We again let $c_{i,\vec s}[m]$, for $0 \leq m \leq n-1$, represent the $(m+1)$th symbol of the codeword associated to the submessage of stream $\vec s \in \Delta_N$.
At time $t \in \{1,...,n\}$, source $S_i$ will thus transmit 
\aln{
X_{S_i}[t] = \ge \sum_{\vec s \in \Delta_N}  T_{\vec s} \; c_{i,\vec s} [t] 
} 
where $\ge = \beta P^{\frac{d-1+2\ep}{2(d+\ep)}}$, for a contant $\beta$ to be determined.
Since the maximum power of a transmit signal from $S_i$ can be loosely upper bounded by
\aln{
 \beta^2 P^{\frac{d-1+2\ep}{d+\ep}} \left( \sum_{\vec s \in \Delta_N} |T_{\vec s}| \right)^2 \hspace{-1.4mm} \gamma^2 P^{\frac{1-\ep}{d+\ep}} & = \beta^2 \hspace{-1mm}  \left( \sum_{\vec s \in \Delta_N} |T_{\vec s}| \right)^2 \gamma^2 P,
}
for any value of $\gamma$ and $N$, we can choose the constant $\beta$ such that the maximum transmit power at the sources is no more than $P$.

\vspace{4mm}

{\noindent \bf Relaying operations:} 

\vspace{2mm}

\noindent Just as in Section \ref{descriptiontv}, the received signals are given by (\ref{recnotaligned0}) and (\ref{recaligned0}), which, in the case of constant channels become respectively
\al{ 
& Y_{V_j}[t] = \ge \sum_{\vec s \in \Delta_{N}}  T_{\vec s} \left( \sum_{i=1}^K h_{S_i,V_j} c_{i,\vec s}[t] \right) + Z_{V_j}[t]. \label{recnotaligned} \\
& Y_{V_j}[t] = \ge \sum_{\vec s \in \Delta_{N+1}}  T_{\vec s} \; u_{j,\vec s} [t] + Z_{V_j}[t]. \label{recaligned}
}
However, since in this case the code symbols $c_{i,\vec s}$, and consequently the $u_{j,\vec s}$s, are integers, it makes sense to consider the  (noiseless) received constellation at each relay, given by
\al{ 
\V = & \left\{ \ge \sum_{\vec s \in \Delta_{N+1}}  T_{\vec s} \; u_{\vec s} : u_{\vec s} \in 
\Z \cap \left[ -K \gamma P^{\frac{1-\ep}{2(d+\ep)}} , K \gamma P^{\frac{1-\ep}{2(d+\ep)}} \right], \forall \; \vec s \in \Delta_{N+1} \right\}. \label{relayconstellation}
}
Each relay $V_j$ will map its received signal $Y_{V_j}[t]$ to the nearest point in $\V$.
This point can then be used to obtain the integers $u_{j,\vec s}$, for $\vec s \in \Delta_{N+1}$, due to the following claim (which is later proven in Section \ref{analysissec}).
\begin{claim} \label{claim1b}
There exists a one-to-one map between points $v \in \V$ and tuples of integers $(u_{\vec s} : \vec s \in \Delta_{N+1})$ with entries in $\Z \cap \left[ -K  P^{\frac{1-\ep}{2(d+\ep)}} , K  P^{\frac{1-\ep}{2(d+\ep)}} \right]$ such that $v = \ge \sum_{\vec s \in \Delta_{N+1}}  T_{\vec s} \; u_{\vec s}$.
\end{claim}
After decoding $u_{j,\vec s}$, for $\vec s \in \Delta_{N+1}$, using this one-to-one map, relay $V_j$ will re-encode all these integers using new transmit directions, exactly as described in Section \ref{descriptiontv}.
More precisely, the transmit signal of relay $V_j$ at time $t+1$ will be given by
\al{ \label{transaligned} 
X_{V_j}[t+1] = \ge' \sum_{\vec s \in \Delta_{N+1}}  \til T_{\vec s} \; u_{j,\vec s} [t],
}
where $\ge'= \beta' P^{\frac{d-1+2\ep}{2(d+\ep)}}$, and $\beta'$ is chosen so that the output power constraint is satisfied (similar to $\beta$).
The new transmit directions $\til T_{\vec s}$ are defined exactly as before, according to (\ref{bdef}), (\ref{directionsfixed2}) and (\ref{tiltdef}).
Moreover, Claim \ref{claim1} still holds in this case, and the transmit signals can be equivalently written as
\al{ \label{transnotaligned}
X_{V_j}[t+1] = \ge' \sum_{\vec s \in \Delta_{N}}  \til T_{\vec s} \left( \sum_{i=1}^K b_{ij} \, c_{i,\vec s} [t] \right).
}

\vspace{4mm}

{\noindent \bf Decoding at the destinations:} 

\vspace{2mm}

\noindent In order to compute the received signals at the destinations, similar to (\ref{relaytransmitvector}), we first express the transmit signals at time $t$ in vector form, as 
\al{ \label{relaytransmitvector1}
\begin{bmatrix} X_{V_1}[t+1]  \\ \vdots \\ X_{V_K}[t+1] \end{bmatrix} = \gamma' \sum_{\vec s \in \Delta_{N}}  \til T_{\vec s}
\begin{bmatrix} b_{11} &  ... & b_{K1} \\  \vdots & \ddots & \vdots \\ 
b_{1K} & ... & b_{KK} \end{bmatrix}
\begin{bmatrix} c_{1,\vec s}[t] \\ \vdots \\ c_{K,\vec s}[t] \end{bmatrix}.
}
Then, similar to (\ref{matrixrec0}), we can obtain
\al{
\begin{bmatrix} Y_{D_1}[t+1] \\ \vdots \\ Y_{D_K}[t+1] \end{bmatrix} & = 
\begin{bmatrix} h_{V_1,D_1} & ... & h_{V_K,D_1} \\ \vdots & \ddots & \vdots \\ 
h_{V_1,D_K} & ... & h_{V_K,D_K} \end{bmatrix} 
\begin{bmatrix} X_{V_1}[t+1] \\ \vdots \\ X_{V_K}[t+1] \end{bmatrix} + \begin{bmatrix} Z_{D_1}[t+1] \\ \vdots \\ Z_{D_K}[t+1] \end{bmatrix} \nonumber \\ 
& = 
\begin{bmatrix} b_{11} & ... & b_{K1} \\ \vdots & \ddots & \vdots \\ 
b_{1K} & ... & b_{KK} \end{bmatrix}^{-1}
\begin{bmatrix} X_{V_1}[t+1] \\ \vdots \\ X_{V_K}[t+1] \end{bmatrix} + \begin{bmatrix} Z_{D_1}[t+1] \\ \vdots \\ Z_{D_K}[t+1] \end{bmatrix} \nonumber \\
& = \gamma' \sum_{\vec s \in \Delta_{N}}  \til T_{\vec s}[t]
\begin{bmatrix} c_{1,\vec s}[t] \\ \vdots \\ c_{K,\vec s}[t] \end{bmatrix} +\begin{bmatrix} Z_{D_1}[t+1] \\ \vdots \\ Z_{D_K}[t+1] \end{bmatrix}. \label{matrixrec}
}
Thus, the received signal at destination $D_j$ at time $t$ is simply given by
\al{ \label{destreceivedscalar}
Y_{D_j}[t+1] = \gamma' \sum_{\vec s \in \Delta_{N}}  \til T_{\vec s} \; c_{j,\vec s}[t] + Z_{D_j}[t+1].
}
The points in the (noiseless) received constellation at each destination, given by
\al{ \label{destconstellation}
\til \V = \left\{ \ge' \sum_{\vec s \in \Delta_{N}}  \tilde T_{\vec s} \; c_{\vec s} : c_{\vec s} \in  \U, \forall \; \vec s \in \Delta_{N}  \right\},
}
can also be uniquely mapped into tuples of integers due to the following claim.
\begin{claim} \label{claim2}
There exists a one-to-one map between points $v \in \til \V$ and tuples of integers $(c_{\vec s} : \vec s \in \Delta_{N})$ with entries in $\U$ such that $v = \ge' \sum_{\vec s \in \Delta_{N}} \til T_{\vec s} \; c_{\vec s}$.
\end{claim}
\noindent Therefore, at each time $t=2,...,n$, destination $D_i$ will first map its received signal to the nearest point in $\til \V$ and then use the one-to-one map between points in $\til V$ and tuples $(c_{\vec s} : \vec s \in \Delta_N)$ with entries in $\U$ to obtain the $L$ integers $c_{i,\vec s}$ encoded by source $S_i$ at time $t-1$.
At time $n$, destination $D_i$ has decoded $L$ data streams of $n$ integers each (in fact, $n-1$ integers, since the integers encoded by the destination at time $t=n$ do not arrive at the destination within the length-$n$ block), and it applies an individual typicality-based decoder to each of these streams to decode the original source message $W_i$.


\subsection{Performance Analysis of AND for constant channels} \label{analysissec}

\vspace{2mm}

Next we show that AND for constant channels can in fact achieve $K$ degrees of freedom.
In order to do that, we first need to bound the error probability of the hard-decoding operations at the relays and destinations.
In the process of doing that, we prove Claims \ref{claim1b} and \ref{claim2}.

\vspace{4mm}

{\noindent \bf Error probability of relaying operations:} 

\vspace{2mm}

\noindent To bound the error probability of the relaying operations, we need to find a lower bound on the minimum distance between two points in the received constellation $\V$, described in (\ref{relayconstellation}).
Since the directions $T_{\vec s}$, for $\vec s \in \Delta_{N+1}$, are all distinct monomials of the channel gains of the first hop, they can be viewed as analytic functions of $h_{S_i,V_j}$, for $1 \leq i,j \leq K$, that are linearly independent over the reals.
Moreover, the distance between any two points in $\V$ has the form
\aln{
\ge \sum_{\vec s \in \Delta_{N+1}}  T_{\vec s} \; u_{\vec s},
}
where each $u_{\vec s}$ can take values in $\Z \cap \left[ -2 K P^{\frac{1-\ep}{2(d+\ep)}} , 2 K  P^{\frac{1-\ep}{2(d+\ep)}} \right]$.
Thus, we can apply Theorem 5 in \cite{MotahariRealInterference} (see also its subsequent remarks and inequality (8) in particular) to conclude that, for almost all values of the channel gains, there exists a constant $\kappa$, independent of $P$, such that the minimum distance of $\V$ satisfies
\aln{
d_{\min } > \ge \frac{\kappa}{\left(2 K P^{\frac{1-\ep}{2(d+\ep)}}\right)^{|\Delta_{N+1}|-1+\ep}}.
}
By choosing $d = |\Delta_{N+1}| = (N+1)^{K^2}$, we have
\al{ \label{mindist}
d_{\min } > \frac{\kappa \beta P^{\frac{d-1+2\ep}{2(d+\ep)}}}{(2 K)^{d-1+\ep} P^{\frac{(1-\ep)(d-1+\ep)}{2(d+\ep)}}} = 
\frac{\kappa \beta}{(2 K)^{d-1+\ep}} P^{\ep/2}.
}
The fact that the minimum distance between any two points in $\V$ is strictly positive implies that there exists a one-to-one map between points $v \in \V$ and tuples of integers $(u_{\vec s} : \vec s \in \Delta_{N+1})$ with entries in $\Z \cap \left[ -K  P^{\frac{1-\ep}{2(d+\ep)}} , K  P^{\frac{1-\ep}{2(d+\ep)}} \right]$, thus proving Claim \ref{claim1b}.
Therefore, after mapping its received signal to the nearest point in $\V$, relay $V_j$ can in fact decode each $u_{j,\vec s}$, $\vec s \in \Delta_{N+1}$, using this one-to-one map.
This procedure will correctly decode each $u_{j,\vec s}$, provided that $|Z_{V_j}[t]| < d_{\min }/2$, implying that the probability of error for relay $V_j$ is at most
\al{
\Pr(|Z_{V_j}[t]| \geq d_{\min }/2) & = 2 \, Q\left( \frac{d_{\min }}{2 \sigma} \right) \nonumber \\
& \leq \exp \left( - \frac{d_{\min }^2}{8 \sigma^2} \right) \nonumber \\
& = \exp (-\delta P^\ep), \label{errorprob}
}
where $\delta$ is a positive contant that is independent of $P$.

\vspace{4mm}

{\noindent \bf Error probability of symbol decoding at the destinations:} 

\vspace{2mm}

\noindent Similar to what we did for the received signals at the relays, we would like to lower bound the minimum distance between two points in the destinations (noiseless) received constellation $\tilde \V$, given in (\ref{destconstellation}).
The following lemma, whose proof we present in the appendix, allows us to use Theorem 5 from \cite{MotahariRealInterference} as we did before.
\begin{lemma} \label{analyticlemma}
The received directions at the destinations, $\til T_{\vec s}$, for $\vec s \in \Delta_{N}$, are analytic functions of $h_{V_i,D_j}$, $1 \leq i,j \leq K$, that are linearly independent over the reals.
\end{lemma}
Theorem 5 from \cite{MotahariRealInterference} now implies that the minimum distance $\til d_{\min }$ between any two points in $\til \V$ can be lower-bounded as
\aln{
\til d_{\min } > \ge' \frac{\til \kappa}{\left(2  P^{\frac{1-\ep}{2(d+\ep)}}\right)^{|\Delta_{N}|-1+\ep}}.
}
for some constant $\til \kappa$ (which is independent of $P$).
Since $d = |\Delta_{N+1}| > |\Delta_{N}|$, for $P > 1$, we have
\al{ \label{mindist}
\til d_{\min } > \frac{\til \kappa \beta' P^{\frac{d-1+2\ep}{2(d+\ep)}}}{2^{|\Delta_N|-1+\ep} P^{\frac{(1-\ep)(d-1+\ep)}{2(d+\ep)}}} = 
\frac{\tilde \kappa \beta'}{2^{|\Delta_N|-1+\ep}} P^{\ep/2}.
}
The fact that the minimum distance between any two points in $\til \V$ is strictly positive implies that there exists a one-to-one map between points $v \in \til \V$ and tuples of integers $(c_{\vec s} : \vec s \in \Delta_{N})$ with entries in $\Z \cap \left[ - P^{\frac{1-\ep}{2(d+\ep)}} ,  P^{\frac{1-\ep}{2(d+\ep)}} \right]$, thus proving Claim \ref{claim2}.
After mapping its received signal to the nearest point in $\til \V$, destination $D_j$ can in fact decode each $c_{j,\vec s}$, $\vec s \in \Delta_N$, using this one-to-one map.
As in (\ref{errorprob}), the probability that $D_i$ incorrectly decodes these integers (provided that no relay made an error in the previous step) is at most
\al{
\Pr(|Z_{D_j}[t]| \geq \til d_{\min }/2) & = \exp (-\til \delta P^\ep), \label{errorprob2}
}
for some constant $\til \delta > 0$.

\vspace{6mm}


{\noindent \bf Achievable rates:} 

\vspace{2mm}

\noindent To determine the rate of our original codebook, we first notice that each data stream between $S_i$ and $D_i$ effectively creates a discrete memoryless channel with input and output alphabets $\U$ and an error probability which can be upper bounded as
\al{
P_e & \leq 1 - \left(1-  \exp (-\delta P^\ep)\right)^K \left(1- \exp (-\til \delta P^\ep)\right) \nonumber \\
& \leq  1 - \left(1- \exp (-\delta' P^\ep)\right)^{K+1} \nonumber \\
& \leq (K+1) \exp (-\delta' P^\ep), \label{channelerror} 
}
where $\delta' = \min(\delta, \til \delta)$.
This allows us to lower bound the mutual information between the input $U$ and the output $\hat U$ of this channel, for a uniform distribution over the input alphabet. 
Using Fano's inequality, we have
\aln{
I(U;\hat U) & \geq H(U) - H(U|\hat U) \\
& \geq \log|\U| - (1+P_e \log|\U|) \\
& = (1-P_e)\log|\U| - 1 \\
& \geq \left(1-(K+1) \exp(-\delta' P^\ep)\right) \\ 
& \quad \quad \quad \left(\frac{1-\ep}{d+\ep} \frac{\log P}2 + 1 \right) - 1, 
}
and we can achieve rate \aln{R = \left(1-(K+1) \exp(-\delta' P^\ep)\right) \hspace{-1mm} \left(\frac{1-\ep}{d+\ep} \frac{\log P}2 + 1 \right)\hspace{-1mm} - 1} over each data stream, by having our original codebook have $2^{nR}$ codewords.
This means that each data stream can achieve 
\aln{
\lim_{P \to \infty} \frac{R}{\tfrac12 \log P} = \frac{1-\ep}{d+\ep} = \frac{1-\ep}{(N+1)^{K^2}+\ep}
}
degrees of freedom.
Since each source transmits $L = |\Delta_{N}| = N^{K^2}$ data streams, each source-destination pair achieves a total of 
\aln{
\frac{(1-\ep)N^{K^2}}{(N+1)^{K^2}+\ep} \geq \frac{(1-\ep)N^{K^2}}{(1+\ep)(N+1)^{K^2}} = \frac{1-\ep}{1+\ep}\left( \frac{N}{N+1} \right)^{K^2}
}
degrees of freedom, for any large $N$ and any small $\ep > 0$, implying that each source-destination pair can achieve arbitrarily close to one degree of freedom.
We conclude that the aligned network diagonalization scheme can achieve arbitrarily close to $K$ degrees of freedom for almost all values of the channel gains, which proves Theorem \ref{thmc}.

\section{Two-Hop Networks with MIMO Nodes}

In this section, we use the result from Theorem \ref{thmtv} in order to characterize the degrees of freedom of two-hop networks where we still have full connectivity at each hop, but each node (sources, relays and destinations) is allowed to have multiple antennas.
In general, we want to focus on a $K \times A \times K$ network, where each node $i \in V$ has $M_i$ (full-duplex) antennas.
For simplicity of exposition,  we will focus on the case of time-varying channels.
However, it should be clear that the results in this section can also be obtained in the case of constant channels, by extending Theorem \ref{thmc} instead.

It is obvious that, in this setting, for certain choices of the number of antennas at each node, it may not be optimal to assign the same number of degrees of freedom to each source-destination pair, as was the case when $M_i = 1$ for all $i$.
Therefore, in this section, instead of focusing on the sum degrees of freedom, we will instead consider the degrees-of-freedom region.

\begin{definition} \label{defn_region}
The degrees-of-freedom region of a $K \times A \times K$ wireless network
is given by
\begin{align}
\vspace{2mm}
\D &= \left\{(d_1,...,d_K) \in \R^K_+ \st \forall \, w_1,...,w_K \in \R_+, w_1 d_1 +...+ w_K d_K \right. \nonumber \\ 
& \left. 
\leq \lim_{P \goesto \infty}\left( \sup_{(R_1,...,R_K) \in C(P)}{\frac{w_1 R_1 + ...+ w_K R_K}{\frac12 \log P}}\right) \right\}.
\end{align} 
\end{definition}

While the formal definition is technical, the degrees-of-freedom region can be intuitively understood as a high-SNR approximation to the capacity region, scaled down by $\frac12 \log P$. 
The sum degrees of freedom $d_{\Sigma}$ from Definition \ref{dofdefn} is simply the point in $\D$ that maximizes the (unweighted) sum of its components.
For two-hop networks with MIMO nodes we then have the following result.

\begin{theorem}
For a $K \times A \times K$ wireless network with time-varying channels where each node $i$ has $M_i$ antennas, the degrees-of-freedom region comprises all nonnegative $K$-tuples $(d_1,...,d_K)$ satisfying
\al{
& \sum_{i=1}^K d_i \leq M_{\text{relays}} \label{thm3a} \\
& d_i \leq \min \left[ M_{S_i}, M_{D_i} \right], \quad \text{ for $i=1,...,K$}, \label{thm3b}
}
where $M_{\text{relays}} = \sum_{i=1}^A M_{V_i}$ is total number of antennas at the relays.
\end{theorem}

Once again, the converse part of this Theorem is obtained from the cut-set bound in a straigthforward manner.
Moreover, given Theorem \ref{thmtv}, the achievability is also easily obtained.
More precisely, for any degrees-of-freedom tuple $(d_1,...,d_K)$ satisfying (\ref{thm3a}), we first discard $M_{\text{relays}} - \sum_{i=1}^K d_i$ out of the total relay antennas. 
Moreover, since $d_i \leq \min \left[ M_{S_i}, M_{D_i} \right]$ from (\ref{thm3b}), we can discard $M_{S_i} - d_i$ out of the source antennas and $M_{D_i} - d_i$ out of the destination antennas.
Then, if we view all remaining antennas as separate nodes, we obtain a $K' \times K' \times K'$ wireless network with time-varying channels, where $K' = \sum_{i=1}^K d_i$.
It is then clear that, by applying Theorem \ref{thmtv}, we can achieve $\sum_{i=1}^K d_i$ sum degrees of freedom on this network, which corresponds to the degrees-of-freedom tuple $(d_1,...,d_K)$ in the original network.

The most interesting aspect of this result is the fact that the cooperation that is allowed among the antennas due to the MIMO setting does not improve the degrees of freedom that can be achieved in the case that all the antennas are viewed as separate nodes.
Notice, however, that this cooperation among antennas may have the power to simplify the schemes we presented in Sections \ref{descriptiontv} and \ref{descriptionc}.


\newpage

\section{Conclusion}

In this work, we showed that, $\kkk$ wireless networks have $K$ degrees of freedom both in the case of time-varying channel coefficients and in the case of constant channel coefficients (in which case the result holds for almost all values of constant channel coefficients).
This result is surprising due to the fact that, in a $\kkk$ wireless network, each destination is subject to interference originated at $K-1$ sources.
Thus, the total number of interference signals that need to be neutralized for $K$ degrees of freedom to be achieved is $O(K^2)$, while the number of variables under our control (i.e., the encoding rules at the sources and the relaying operations) is only $O(K)$.
Moreover, this result answers the important conceptual question of whether global interference management can significantly outperform hop-by-hop interference management.
As we showed, global interference management techniques have the potential to effectively remove all the interference experienced by the destinations in two-hop wireless networks.
The same objective cannot be achieved if we restrict ourselves to hop-by-hop interference management.


Our main result is proven through the introduction of a new coding strategy, which we call Aligned Network Diagonalization.
In the case of time-varying channels in particular, the scheme can be seen as providing an affirmative answer to a linear diagonalization question: is it possible to apply linear transformations at the sources, relays and destinations, so that the overall linear tranformation of the network is diagonal?
The main idea of the scheme lies in the operations performed by the relays.
These operations can be understood as modifying the received signals at the relays so that it ``looks like'' the transfer matrix of the first hop is the inverse of the transfer matrix of the second hop.
This way, we can effectively diagonalize the network, creating parallel interference-free channels from each source to its corresponding destination. 
These interference-free channels allow each source-destination pair to achieve arbitrarily close to one degree of freedom.

While our results imply the tightness of the cut-set bound, it is important to point out that this is likely to be the case only for degrees of freedom. 
For example, this occurs for the two-user interference channel, where the degrees-of-freedom cut-set bound of $1$ is trivially tight, but, as shown in \cite{ETW}, tighter outer bounds can be obtained through genie-aided arguments and by considering distinct interference regimes.
Thus, one would expect that more sophisticated outer bounds can be developed for $\kkk$ wireless network as well.
In this sense, a promising direction is to consider a deterministic model of the $\kkk$ wireless network.
Deterministic models of wireless networks have been proven useful in the study of the capacity of both multi-hop single-flow networks \cite{ADTJ09} and single-hop multi-flow networks \cite{ElGamalCosta,BreslerTseEuro}.
Not only do they usually provide new insights about the original stochastic problem, but they can in fact be shown, in several cases, to approximate well the capacity of their non-deterministic (usually AWGN) counterparts.
A step towards studying $\kkk$ wireless networks under deterministic models was taken in \cite{deterministickkk}.
By using the worst-case noise result from \cite{wcnoisefull}, it was shown that the capacity region of an AWGN $\kkk$ wireless network is a subset of the capacity region of the same network under the truncated deterministic model \cite{ADTJ09} (where nodes are given slightly more power). 
This fact is particularly interesting because it allows us to look for outer bounds on the capacity region of the AWGN $\kkk$ wireless network by focusing on the truncated deterministic channel model, which is expected to reveal  combinatorial structures of the problem that are not apparent in the AWGN setting.

Other directions for future work include a study of the time expansion required for our proposed scheme to be performed.
In particular, we notice that the linear version of AND relies on the fast variation of the channel gains in the network and requires a large number of distinct channel realization in order to achieve close to one degree of freedom per user.
If we limit the available time (and space) diversity, as considered for instance in \cite{bresler3user}, it is not clear if the same gains achieved by AND can be obtained. 
Of particular interest is the tradeoff between achievable degrees of freedom and the amount of channel diversity available.
Furthermore, one would expect this tradeoff to be strongly affected by the addition of extra relays to the network.
Thus we would like to know if AND can be generalized to take advantage of a larger number of relays, which can provide more spatial diversity, in order to reduce its requirement for time diversity.

\section{Acknowledgements}
The research of A.~S.~Avestimehr and I.~Shomorony was supported in part by the NSF CAREER award 0953117, NSF Grants CCF-1144000 and CCF-1161720, a research grant from Samsung Advanced Institute of Technology (SAIT), and the U.S. Air Force Young Investigator Program award FA9550-11-1-0064.

\newpage

\appendices \label{appsection}




\section{Proof of Lemma \ref{invertiblelemma}} \label{appinvertible}
%

\begin{lemmarep}{\ref{invertiblelemma}}
Let $\vec p (x_1,...,x_K) = \left[ p_1 (x_1,...,x_K),...,  p_d (x_1,...,x_K) \right]^\dagger$, where each $p_i(x_1,...,x_K)$ is a distinct monomial on the variables $x_1,...,x_K$.
Then, the determinant of the $d \times d$ matrix
\aln{
\left[ \vec p (x_{1,1},...,x_{1,K}), \vec p (x_{2,1},...,x_{2,K}), ..., \vec p (x_{d,1},...,x_{d,K}) \right]
}
is a non-identically zero polynomial on the variables $x_{1,1},...,x_{1,K},...,x_{d,1},...,x_{d,K}$.
\end{lemmarep}

\vspace{2mm}

\begin{proof}
Obviously, the determinant of $\left[ \vec p (x_{1,1},...,x_{1,K}), \vec p (x_{2,1},...,x_{2,K}), ..., \vec p (x_{d,1},...,x_{d,K}) \right]$ is a polynomial on the variables $x_{1,1},...,x_{1,K},...,x_{d,1},...,x_{d,K}$.
To show that it is non-identically zero, we just need to show that, for some choice of $x_{1,1},...,x_{1,K},...,x_{d,1},...,x_{d,K}$, the determinant is nonzero.
We do this by showing inductively that we can choose $x_{1,1},...,x_{1,K}$, then $x_{2,1},...,x_{2,K}$ and so on, so that, when we choose $x_{j,1},...,x_{j,K}$, the column $\vec p (x_{j,1},...,x_{j,K})$ is linearly independent from $\vec p (x_{1,1},...,x_{1,K}),...,\vec p (x_{j-1,1},...,x_{j-1,K})$.
The base case is trivial.
Fix any $j \in \{2,...,K\}$, and suppose $x_{1,1},...,x_{1,K},...,x_{j-1,1},...,x_{j-1,K}$ have been chosen such that the linear space spanned by $\vec p (x_{1,1},...,x_{1,K}),...,\vec p (x_{j-1,1},...,x_{j-1,K})$, $L$, has dimension $j-1$.
Since $j-1 < d$, there must be constants $\alpha_1,...,\alpha_d$ (not all zero) such that, for any $(y_1,..,y_d) \in L$, $\sum_{i=1}^d \alpha_i y_i = 0$.
But since each $p_i(x_1,...,x_K)$ is a distinct monomial on the variables $x_1,...,x_K$, $\sum_{i=1}^d \alpha_i p_i(x_1,...,x_K)$ is not identically zero.
Thus, we can choose $x_{j,1},...,x_{j,K}$, such that $\sum_{i=1}^d \alpha_i p_i(x_{j,1},...,x_{j,K}) \ne 0$, which implies that $\vec p (x_{j,1},...,x_{j,K}) \notin L$, completing the proof.
\end{proof}

\section{Proof of Lemma \ref{analyticlemma}}

\begin{lemmarep}{\ref{analyticlemma}}
The received directions at the destinations, $\til T_{\vec s}$, for $\vec s \in \Delta_{N}$, are analytic functions of $h_{V_i,D_j}$, $1 \leq i,j \leq K$, that are linearly independent over the reals.
\end{lemmarep}

\begin{proof}
To prove that each $\til T_{\vec s}$, for $\vec s \in \Delta_{N}$, is an analytic function of $h_{V_i,D_j}$, $1 \leq i,j \leq K$, we notice that
if we let
\aln{
H = \begin{bmatrix} h_{V_1,D_1} & h_{V_2,D_1} & ... & h_{V_K,D_1} \\ h_{V_1,D_2} & h_{V_2,D_2} & ... & h_{V_K,D_2} \\ \vdots & \vdots & \ddots & \vdots \\ 
h_{V_1,D_K} & h_{V_2,D_K} & ... & h_{V_K,D_K} \end{bmatrix},
}
then, for $1 \leq i,j \leq K$ we can write $b_{ij} = \frac{C_{ij}}{\det H}$, where $C_{ij}$ is the cofactor of the $(i,j)$ entry of $H$.
This means that each $b_{ij}$ is a ratio of two polynomials with $h_{V_i,D_j}$, $1 \leq i,j \leq K$, as variables.
Since each $\til T_{\vec s}$ is a distinct monomial of the $b_{ij}$s, it is clear that each $\til T_{\vec s}$ is an analytic function of $h_{V_i,D_j}$, $1 \leq i,j \leq K$.

Next, suppose by contradiction that $\til T_{\vec s}$, for $\vec s \in \Delta_N$, are not linearly independent over the reals.
Then there must be real numbers $\alpha_{\vec s}$, for $\vec s \in \Delta_N$, not all zero, such that 
\aln{
\sum_{\vec s \in \Delta_N} \alpha_{\vec s} \til T_{\vec s} = 0
}
for all values of $h_{V_i,D_j}$, for $1 \leq i,j \leq K$.
However, since the $\til T_{\vec s}$, for $\vec s \in \Delta_N$ are distinct monomials of the $b_{ij}$s, we have that, for almost all values of the $b_{ij}$s, $\sum_{\vec s \in \Delta_N} \alpha_{\vec s} \til T_{\vec s} \ne 0$.
Since for almost all values of the $b_{ij}$s, the matrix
\aln{
B = \begin{bmatrix} b_{11} & b_{21} & ... & b_{K1}  \\ b_{12} & b_{22} & ... & b_{K2} 
\\ \vdots & \vdots & \ddots & \vdots \\ 
b_{1K} & b_{2K} & ... & b_{KK} \end{bmatrix}
}
is invertible, we can find $b_{11}, b_{12},..., b_{KK}$ for which $B$ is invertible and $\sum_{\vec s \in \Delta_N} \alpha_{\vec s} \til T_{\vec s} \ne 0$ (with the $\til T_{\vec s}$\,s seen as functions of the $b_{ij}$s).
But this means that if we choose the values of $h_{V_i,D_j}$, for $1 \leq i,j \leq K$, by setting $H = B^{-1}$, we will have $\sum_{\vec s \in \Delta_N} \alpha_{\vec s} \til T_{\vec s} \ne 0$ (with the $\til T_{\vec s}$\,s seen as functions of the $h_{V_i,D_j}$s), which is a contradiction.
\end{proof}

\bibliographystyle{unsrt}



\end{document}